\documentclass[conference]{IEEEtran}
\pdfoutput=1
\IEEEoverridecommandlockouts
\usepackage{cite}
\usepackage{amsmath,amssymb,amsfonts}
\usepackage{amsthm}
\usepackage{algorithmic}
\usepackage{graphicx}
\usepackage{epstopdf}
\usepackage{textcomp}
\usepackage{xcolor}
\usepackage{setspace}

\def\BibTeX{{\rm B\kern-.05em{\sc i\kern-.025em b}\kern-.08em
    T\kern-.1667em\lower.7ex\hbox{E}\kern-.125emX}}

\theoremstyle{definition}
\newtheorem*{defn*}{\protect\definitionname}
 \theoremstyle{plain}
 \newtheorem*{thm*}{\protect\theoremname}
 \theoremstyle{definition}
 \newtheorem{defn}{\protect\definitionname}
\theoremstyle{plain}
\newtheorem{thm}{\protect\theoremname}

\makeatother

 \providecommand{\definitionname}{Definition}
 \providecommand{\theoremname}{Theorem}
\providecommand{\theoremname}{Theorem}

\usepackage[unicode=true]{hyperref}

\begin{document}

\title{Generalized Random Surfer-Pair Models
}

\author{\IEEEauthorblockN{Sai Kiran Narayanaswami}
\IEEEauthorblockA{\begin{minipage}{0.5\columnwidth}
  \begin{center}
    \textit{Robert Bosch Centre for Data Science and AI,}
  \end{center}
\end{minipage}\\\textit{Indian Institute of Technology Madras} \\
Chennai, India \\
saikirann94@gmail.com}
\and
\IEEEauthorblockN{Balaraman Ravindran}
\IEEEauthorblockA{\textit{Dept. of Computer Science and Engineering}\\
\begin{minipage}{0.5\columnwidth}
  \begin{center}
    \textit{Robert Bosch Centre for Data Science and AI,}
  \end{center}
\end{minipage}\\
\textit{Indian Institute of Technology Madras}\\
Chennai, India \\
ravi@cse.iitm.ac.in}
\and
\IEEEauthorblockN{Venkatesh Ramaiyan}
\IEEEauthorblockA{\textit{Dept. of Electrical Engineering} \\
\begin{minipage}{0.5\columnwidth}
  \begin{center}
    \textit{Robert Bosch Centre for Data Science and AI,}
  \end{center}
\end{minipage}\\
\textit{Indian Institute of Technology Madras}\\
Chennai, India \\
rvenkat@ee.iitm.ac.in}
}

\maketitle

\begin{abstract}
SimRank is a widely studied link-based similarity measure that is
known for its simple, yet powerful philosophy that two nodes are similar
if they are referenced by similar nodes. While this philosophy has
been the basis of several improvements, there is another useful, albeit
less frequently discussed interpretation for SimRank known as the
Random Surfer-Pair Model. In this work, we show that other well known
measures related to SimRank can also be reinterpreted using Random
Surfer-Pair Models, and establish a mathematically sound, general
and unifying framework for several link-based similarity measures.
This also serves to provide new insights into their functioning and
allows for using these measures in a Monte Carlo framework, which
provides several computational benefits. Next, we describe how the
framework can be used as a versatile tool for developing measures
according to given design requirements. As an illustration of this
utility, we develop a new measure by combining the benefits of two
existing measures under this framework, and empirically demonstrate
that it results in a better performing measure.
\end{abstract}

\section{Introduction}

Advances in computing and information systems have enabled the collection
of data on a large scale. A sizable portion is available in network
form, with some well known examples being social networks and the
World Wide Web. Naturally, the need arises to effectively utilize
this data in the relevant domains. One area where network data has
enormous potential is in recommender systems.

With the advent of electronic commerce, data is often available for
products that customers have purchased, which can be used for collaborative
filtering to provide recommendations for further purchases. As human
knowledge grows, leading to rapidly expanding bodies of literature,
there is an increasing requirement for effective recommender systems
to aid researchers. Thus, utilizing bibliographic network data such
as citation and co-authorship networks to provide recommendations
is of rising importance.

In general, some form of similarity assessment would be a key component
of recommender systems. We consider one class of similarity measures
that work with only the link structure of the network known as \emph{structural}
(or \emph{link-based}) similarity measures. Among the first such well
known measures are Co-Citation \cite{cocitation} and its counterpart
Bibliographic coupling \cite{bibcoupling} that measure respectively
the frequency with which nodes refer to two given nodes , and the
frequency with which nodes are referenced together by two given nodes.
Amsler \cite{amsler1972} is a combination of the former two measures.
SimRank \cite{simrank} with its intuitive graph-theoretic model was
a breakthrough which formed the basis of several subsequent measures.
Among the notable ones are P-Rank \cite{Zhao:2009:PCS:1645953.1646025},
PSimRank \cite{psimrank} and SimRank{*} \cite{simrank-star}.

At the focus of this work is the probabilistic interpretation of SimRank
known as the Random Surfer-Pair (RSP) Model \cite{simrank}. We develop
a generalization of this interpretation that applies to several measures
and present a unified view of these measures. Casting them under the
Generalized Random Surfer-Pair (GRSP) model also provides new insights
into their functioning: with P-Rank, the GRSP model brings to light
an otherwise indiscernible peculiarity, and for SimRank{*} provides
an elegant and intuitive justification. We show how the framework
also provides the tools necessary for performing Monte Carlo computation
of these measures and discuss the computational benefits that result.
Next, we outline the ways in which the GRSP model can be used to design
measures, and the considerations involved in its usage for developing
new measures. We then apply the framework to develop a hybrid measure,
PSimRank{*} to combine the benefits of two existing measures, PSimRank
and SimRank{*}. Experiments are performed to demonstrate empirically
how this combination improves upon both of the measures on which it
is based in terms of retrieval efficiency on a large (more than 2
million nodes) citation network from the Arnetminer dataset.

\section{Background}

In this section, we present SimRank and other relevant link-based
similarity measures, and then the RSP interpretation for SimRank.
Throughout, we consider networks with directed, simple graphs. We
denote the graph by $G$, its vertex set by $V$, and its edge set
by $E$. For any vertex $a$, $I\left(a\right)$ denotes its in-neighbors
and $O\left(a\right)$ its out-neighbors. We index into these sets
as $I_{i}\left(a\right)$ and $O_{i}\left(a\right)$.

\subsection{SimRank and related measures}

The famous SimRank philosophy that two nodes are similar if they are
referenced by similar nodes is formalized recursively as:
\begin{equation}
s\left(a,b\right)=\frac{C}{|I(a)||I(b)|}\sum_{i=1}^{|I(a)|}\sum_{j=1}^{|I(b)|}s\left(I_{i}(a),I_{j}(b)\right)\label{eq:sr_rec}
\end{equation}

where $a\neq b$ and $C<1$ is a fixed, positive parameter. In the
rare cases where either of the nodes $a$ and $b$ have no in-neighbors,
the similarity is considered to be zero. Maximal self similarity applies,
with $s\left(x,x\right)=1$ $\forall x\in V$. These two base cases
also apply for other measures discussed subsequently and nodes $a$
and $b$ on the left hand side are assumed to be distinct except as
specified otherwise.

An out-link variant of SimRank, rvs-SimRank is also described in \cite{struct-sim-survey},
which has the following form:
\begin{equation}
s\left(a,b\right)=\frac{C}{|O(a)||O(b)|}\sum_{i=1}^{|O(a)|}\sum_{j=1}^{|O(b)|}s\left(O_{i}(a),O_{j}(b)\right)\label{eq:rvs_sr_rec}
\end{equation}

Some deficiencies that have been identified in SimRank are as follows
\cite{struct-sim-survey} :
\begin{itemize}
\item The In Links Consideration Problem : SimRank is unavailable (i.e set
to zero) when either node has no in-neighbors, even though there may
be evidence of similarity in the out-neighbors.
\item The Pairwise Normalization Problem : This is the counter-intuitive
effect that the SimRank score of a pair of nodes can \emph{decrease}
as there are more and more nodes referring to both of them. Consider
two nodes $u$ and $v$ having several (say $k$) nodes that refer
to both of them. Now, if these nodes are unrelated to each other (i.e
they have zero similarity), the SimRank score between $u$ and $v$
is found to be $\frac{C}{k}$, which \emph{decreases} with $k$. Thus,
even though there are more witnesses to the similarity of $u$ and
$v$, their SimRank score decreases.
\item The Level-wise computation problem : It can be shown that SimRank
is unavailable for node pairs that don't have any paths of \emph{equal}
length to a common node. That is, it discards any evidence of similarity
provided by path pairs of unequal length to a common node.
\end{itemize}
We now present some measures that were developed to address these
issues.

\subsubsection{P-Rank}

P-Rank (\cite{Zhao:2009:PCS:1645953.1646025}) was proposed to take
into account out-links as well in computing similarity. It has the
following recursive form :
\begin{multline}
s\left(a,b\right)=\lambda\times\frac{C}{|I(a)||I(b)|}\sum_{i=1}^{|I(a)|}\sum_{j=1}^{|I(b)|}s\left(I_{i}(a),I_{j}(b)\right)\\
+(1-\lambda)\times\frac{C}{|O(a)||O(b)|}\sum_{i=1}^{|O(a)|}\sum_{j=1}^{|O(b)|}s\left(O_{i}(a),O_{j}(b)\right)\label{eq:prank_rec}
\end{multline}

It essentially adds an additional clause to the SimRank philosophy:
\emph{two entities are also similar if they reference similar entities}.

The base cases are similar to those of SimRank, except that only the
term that corresponds to in (out) neighbors gets zeroed out if one
or both of $\left(a,b\right)$ doesn't have in (out) neighbors. This
ensures that P-Rank is not unavailable for node pairs without in-links
as is the case with SimRank, as long as they have out links.

\subsubsection{PSimRank}

The recursive form of PSimRank (\cite{psimrank}) is as follows :
\begin{align}
s\left(a,b\right)= & \frac{C|I\left(a\right)\cap I\left(b\right)|}{|I\left(a\right)\cup I\left(b\right)|}\cdot1\nonumber \\
+ & \frac{C}{|I\left(a\right)\cup I\left(b\right)||I\left(b\right)|}\sum_{\substack{a^{'}\in I\left(a\right)\setminus I\left(b\right)\\
b^{'}\in I\left(b\right)
}
}s\left(a^{'},b^{'}\right)\label{eq:psim_rec}\\
+ & \frac{C}{|I\left(a\right)\cup I\left(b\right)||I\left(a\right)|}\sum_{\substack{b^{'}\in I\left(b\right)\setminus I\left(a\right)\\
a^{'}\in I\left(a\right)
}
}s\left(a^{'},b^{'}\right)\nonumber 
\end{align}

PSimRank solves the pairwise normalization problem by assigning greater
importance to common in-neighbors. It is evident that unlike SimRank
where each term $s\left(a,b\right)$ has the same weight, the weights
have been redistributed so that terms of the form $s\left(x,x\right)$
(which are always 1, and make up the constant term) are given more
weight than all other terms.

\subsubsection{SimRank{*}}

SimRank{*} (\cite{simrank-star}) was proposed to solve the level-wise
computation problem of SimRank and has the following recursive form:
\begin{equation}
s\left(a,b\right)=\frac{C}{2|I(a)|}\sum_{i=1}^{|I(a)|}s\left(I_{i}(a),b\right)+\frac{C}{2|I(b)|}\sum_{i=1}^{|I(b)|}s\left(a,I_{i}\left(b\right)\right)\label{eq:simrank-star-rec}
\end{equation}

The measure is derived in \cite{simrank-star} by actually enumerating
all pairs of paths of (possibly) unequal length from $a$ and $b$
to a common node and computing the weighted sum of an exponentially
decayed score associated with each path. Later on, we present a much
simpler explanation as to how it works under the GRSP interpretation.

\subsection{The RSP Model for SimRank}

The Random Surfer-Pair interpretation for Simrank is based on a random
experiment involving two random walks (or surfers) starting at the
given nodes $a$ and $b$, and traversing the graph \emph{backwards}
until they meet. That is, at the end of each step, each walk transitions
to a randomly chosen in-neighbor. If either of the current nodes have
no in-neighbors, the experiment is stopped.
\begin{defn*}
Let $L(a,b)$ be the random variable that gives the number of steps
taken until the surfers meet starting from $a$ and $b$ respectively.
The \emph{expected }$f$\emph{-meeting distance} between $a$ and
$b$ is defined for a given function $f$ as $\mathbb{E}\left[f\left(L\left(a,b\right)\right)\right]$.

The $f$-meeting distance can be viewed as a score resulting from
each instance of the experiment, and is itself a random variable.
It turns out that for a specific choice of $f$, the expected score
is nothing but the SimRank of $\left(a,b\right)$. This equivalence
of the Random Surfer-Pair formulation and the recursive form in equation
\ref{eq:sr_rec} are proved in \cite{simrank} :
\end{defn*}
\begin{thm*}
SimRank as defined by Equation (\ref{eq:sr_rec}) is the same as the
expected $f$\emph{-meeting distance} between $a$ and $b$ for $f(t)=C^{t}$,
where $0<C<1$.
\end{thm*}
If the experiment is stopped because of unavailability of neighbors,
$L\left(a,b\right)$ is considered to be infinite, thus making the
$f$-meeting distance zero for that run of the experiment. Also, the
base case of maximal self-similarity follows naturally because if
$a=b$, the surfers deterministically meet at time $t=0$, giving
a score of 1 always.

In this interpretation, two nodes are similar if they are close to
some source(s) of similarity. As we will discuss, various measures
differ in what they consider to be sources of similarity and how much
weight is assigned to them.

\section{Related Work}

An RSP interpretation has already been made for PSimRank, and it was
in fact the way in which the measure was developed\cite{psimrank}.
PSimRank is defined the same way as the RSP model for SimRank, except
that the transition probabilities for the surfers are modified so
that they are more likely to meet at a common in-neighbor in the next
step based on the number of such common in-neighbors at their current
positions.

The improvements to Simrank are justified using this RSP model, and
the recursive form in Eqn. \ref{eq:psim_rec} is also derived from
it. However, unlike this work, the RSP model is not discussed as a
general framework, and no general connections between the recursive
forms and the RSP model are established.

We will arrive at the same RSP model later when we apply the GRSP
model to PSimRank in Section \ref{subsec:PSimRank-grsp}.

\section{Generalizing the RSP Model}

For defining the GRSP model, we treat the Random Surfer-Pair experiment
as a single random walk, but on a compound state space that consists
of vertex pairs from $V\times V$ \cite{simrank} to indicate the
positions of both surfers, and also a ``stopped'' state, which represents
unavailability. We use the letter $h$ to denote a typical state from
this space, which we denote by $\mathcal{S}$. 

The transition probabilities for this random walk are denoted as $p(h^{'}\,|\,h)$,
the probability of transitioning to $h^{'}$ from $h$. The stopped
state is an absorbing state, that is once the state is reached, it
is impossible to leave. We can collect these probabilities into a
matrix $\mathbf{P}$.

The idea is that different measures can be realized for different
choices of transition probabilities, formalized by the following definition
of the GRSP model :
\begin{defn}
\label{def:gen_rsp_model}For a particular matrix of transition probabilities
$\mathbf{P}$, consider the following combined random walk experiment
over the compound state space $\mathcal{S}$ starting from $\left(a,b\right)$
at time $t=0$:

\begin{itemize}
\item If the current state of the walk is $h$, the walk moves to the next
state with probability $p(h^{'}\,|\,h)\,\forall\,h^{'}\in\mathcal{S}$
as specified by $\mathbf{P}$.
\item The walk ends when a state of the form $\left(x,x\right)$ is reached,
for some $x\in V$.
\end{itemize}
Let $L(a,b)$ be the random variable that gives the number of steps
taken until the walk ends. The \emph{expected }$f$\emph{-meeting
distance} between $a$ and $b$ is defined for this combined walk
for the function $f(t)=C^{t}$ as $\mathbb{E}\left[f\left(L\left(a,b\right)\right)\right]$
with $0<C<1$. This is a function of $\left(a,b\right)$ which we
call the \emph{similarity measure induced by} $\mathbf{P}$ under
the Generalized Random Surfer-Pair model.
\end{defn}
One can think of other choices for $f$ to encode different extents
of decay of similarity, such as $f(t)=\frac{C}{t^{2}}$. However,
the specific choice of $f(t)=C^{t}$ is what leads to the useful properties
we discuss below.

The termination condition is equivalent to the random surfers meeting
at some node for the first time. If the walk goes into the stopped
state, it stays there forever, and does not reach a state of the form
$\left(x,x\right)$. This gives an infinite number of steps, thus
leading to a score of zero. Again, the base case of maximal self-similarity
applies here as well.

\section{Equivalence to Recursive Form\label{sec:Equivalence-to-recursive}}

An interesting observation to be made is that the coefficients of
any $s(a^{'},b^{'})$ in the recursive formulation of SimRank (Equation
\ref{eq:sr_rec}) are the same as the transition probabilities for
the surfers in the Random Surfer-Pair Model going from $\left(a,b\right)$
to $(a^{'},b^{'})$. This leads one to believe there could be a similar
relationship for any transition probabilities $\mathbf{P}$. Indeed,
this is true and the results are formally presented in the remainder
of this section.

For a given transition probability matrix $\mathbf{P}$, consider
the following set of recursive equations defined for all node pairs
$\left(a,b\right)$ :
\begin{equation}
s\left(a,b\right)=C\sum_{(a^{'},b^{'})\in\mathcal{R}\left(\left(a,b\right)\right)}p\left((a^{'},b^{'})\,\middle|\,\left(a,b\right)\right)s(a^{'},b^{'})\label{eq:general_rec}
\end{equation}

Here, $\mathcal{R}\left(\left(a,b\right)\right)\subseteq V\times V$
is a region of support (which we will also refer to as \emph{support
set}) for $\left(a,b\right)$ under $\mathbf{P}$, that is where the
transition probability $p\left((a^{'},b^{'})\,\middle|\,\left(a,b\right)\right)$
is non-zero. Note that this \emph{does not} include the stopped state,
which means the sum of the coefficients appearing in the above equation
need not be 1 (of course, they have to be less than 1). 

The same base case of $s\left(a,a\right)=1\,\forall\,a\in V$ is used.
If $\mathcal{R}\left(\left(a,b\right)\right)=\left\{ \phi\right\} $,
$s\left(a,b\right)$ is taken to be zero unless $a=b$. These equations
define what is called the \emph{recursive similarity measure induced
by }$\mathbf{P}$.

An iterative form is also defined:
\begin{equation}
s_{k+1}\left(a,b\right)=C\sum_{(a^{'},b^{'})\in\mathcal{R}\left(\left(a,b\right)\right)}p\left((a^{'},b^{'})\,\middle|\,\left(a,b\right)\right)s_{k}(a^{'},b^{'})\label{eq:general_iter}
\end{equation}

The following results establish the mathematical soundness of the
GRSP model.
\begin{thm}
\label{thm:convergence}The following results hold true for the iterative
form :

\begin{itemize}
\item \textbf{Monotonicity and boundedness} :
\[
0\leq s_{k}\left(a,b\right)\leq s_{k+1}\left(a,b\right)\leq1\qquad\forall\,\left(a,b\right)\in V\times V
\]
\item \textbf{Convergence to limit} : The sequence $s_{k}\left(a,b\right)$
converges to a limit (obviously between 0 and 1 by the previous part)
for all $\left(a,b\right)\in V\times V$.
\end{itemize}
\end{thm}
\begin{proof}
The monotonicity and boundedness are proved by induction. The inductive
hypothesis is that\textbf{
\[
0\leq s_{k-1}\left(a,b\right)\leq s_{k}\left(a,b\right)\leq1\qquad\forall\,\left(a,b\right)\in V\times V
\]
}

The base case of this for $k=0$ is trivial since $s_{0}\left(x,x\right)=1$
and $s_{0}\left(a,b\right)=0$ $\forall\,a\neq b$, and so is the
case with $a=b$ . The inductive step is as follows :

\textbf{Monotonicity :} We have
\begin{multline*}
s_{k+1}\left(a,b\right)-s_{k}\left(a,b\right)=C\times\\
\sum_{(a^{'},b^{'})\in\mathcal{R}\left(\left(a,b\right)\right)}p\left((a^{'},b^{'})\,\middle|\,\left(a,b\right)\right)\left[s_{k}(a^{'},b^{'})-s_{k-1}(a^{'},b^{'})\right]
\end{multline*}

But $s_{k}(a^{'},b^{'})-s_{k-1}(a^{'},b^{'})\geq0$ by the inductive
hypothesis, and $p((a^{'},b^{'})\,|\,\left(a,b\right))\geq0$ since
it is a probability, thus proving the monotonicity.

\textbf{Boundedness :} From the inductive hypothesis, $0\leq s_{k}\left(a,b\right)\leq1$.
Therefore,
\begin{eqnarray*}
s_{k+1}\left(a,b\right) & = & C\sum_{(a^{'},b^{'})\in\mathcal{R}\left(\left(a,b\right)\right)}p\left((a^{'},b^{'})\,\middle|\,\left(a,b\right)\right)s_{k}(a^{'},b^{'})\\
 & \leq & C\sum_{(a^{'},b^{'})\in\mathcal{R}\left(\left(a,b\right)\right)}p\left((a^{'},b^{'})\,\middle|\,\left(a,b\right)\right)\cdot1\\
 & \leq & C\:\leq\:1
\end{eqnarray*}

Where we have used the fact that $\sum_{(a^{'},b^{'})\in\mathcal{R}\left(\left(a,b\right)\right)}p((a^{'},b^{'})\,|\,\left(a,b\right))\leq1$,
since it is a sum of transition probabilities out of $\left(a,b\right)$
(possibly less than one because of the stopped state). Similarly,
it can be shown that $s_{k+1}\left(a,b\right)\geq0$.

\textbf{Convergence :} Since $s_{k}\left(a,b\right)$ is bounded and
non-decreasing, by the Completeness Axiom of Calculus, $s_{k}\left(a,b\right)$
converges to a limit $\forall\,\left(a,b\right)\in V\times V$, which
we denote by $g\left(a,b\right)$. Of course, this limit must be between
0 and 1 as the sequence itself is bounded in that range.
\end{proof}
From the above, the following result follows :
\begin{thm}
\label{thm:uniq}There exists a unique solution to the system of equations
defined by equation \ref{eq:general_rec}.
\end{thm}
\begin{proof}
Let two solutions to Equation \ref{eq:general_rec} be $s_{1}$ and
$s_{2}$. Define their difference

\[
\delta\left(a,b\right)=s_{1}\left(a,b\right)-s_{2}\left(a,b\right)
\]

Let $M$ be the maximum absolute value of $\delta$, that is $\max_{\left(a,b\right)}\left|\delta\left(a,b\right)\right|$.
Let this maximum value be achieved for $\left(a,b\right)$, that is
$M=\left|\delta\left(a,b\right)\right|$. If $a=b$, then clearly
$M=0$ as both $s_{1}$ and $s_{2}$ must satisfy the maximal self-similarity
base case. Otherwise, we have
\begin{multline*}
M=\left|C\sum_{(a^{'},b^{'})\in\mathcal{R}\left(\left(a,b\right)\right)}p\left((a^{'},b^{'})\,\middle|\,\left(a,b\right)\right)\left[s_{1}(a^{'},b^{'})-\right.\right.\\
\Biggl.\left.s_{2}(a^{'},b^{'})\right]\Biggr|\leq C\sum_{(a^{'},b^{'})\in\mathcal{R}\left(\left(a,b\right)\right)}p\left((a^{'},b^{'})\,\middle|\,\left(a,b\right)\right)\left|\delta(a^{'},b^{'})\right|\\
\leq C\sum_{(a^{'},b^{'})\in\mathcal{R}\left(\left(a,b\right)\right)}p\left((a^{'},b^{'})\,\middle|\,\left(a,b\right)\right)M\leq CM
\end{multline*}

Here, we have used the fact that since $\left(a,b\right)$ maximizes
$\left|\delta\left(\cdot,\cdot\right)\right|$, $\left|\delta(a^{'},b^{'})\right|\leq M$
and again the fact that $\sum_{(a^{'},b^{'})\in\mathcal{R}\left(\left(a,b\right)\right)}p((a^{'},b^{'})\,|\,\left(a,b\right))\leq1$.

Now, since $M$ is an absolute value, and $M\leq CM$ with $0<C<1$,
we must have $M=0$. This proves that $s_{1}$and $s_{2}$ are always
the same. Thus, there exists a unique solution to the recursive form
of Equation \ref{eq:general_rec}, and that solution can be obtained
as the limit of the iterative form.
\end{proof}
Which leads to our central result :
\begin{thm}
\label{thm:equivalence}The similarity measure induced by $\mathbf{P}$
according to definition \ref{def:gen_rsp_model} is the same as the
recursive similarity measure induced by $\mathbf{P}$ (Equation \ref{eq:general_rec}).
\end{thm}
\begin{proof}
We first need to show that the expected $f$-meeting distances, for
which we overload the notation $s\left(a,b\right)$ satisfy the recursive
form. Let $\mathcal{W}_{a,b}$ be the set of all compound walks from
$\left(a,b\right)$ to a state of the form $\left(x,x\right)$. Let
$l\left(w\right)$ denote the length of such a walk $w$, and $p\left(w\right)$
the total probability, which is the product of the probabilities of
the individual transitions. Then, by definition of the expected $f$-meeting
distance,
\begin{equation}
s\left(a,b\right)=\sum_{w\in\mathcal{W}_{a,b}}p\left(w\right)C^{l\left(w\right)}\label{eq:fmdproof}
\end{equation}

Now, consider the set of all such compound walks from one step ahead,
that is, $\bigcup_{(a^{'},b^{'})\in\mathcal{R}\left(\left(a,b\right)\right)}\mathcal{W}_{a^{'},b^{'}}$.
Note that all the individual sets $\mathcal{W}_{a^{'},b^{'}}$ are
disjoint. Now, clearly this collection of walks differs from $\mathcal{W}_{a,b}$
only in the inclusion of the first transition to some $(a^{'},b^{'})$.
Therefore, a bijection exists between this set and $\mathcal{W}_{a,b}$,
and so $\mathcal{W}_{a,b}$ can be enumerated in terms of the new
collection. \\
\\
This means that for every member $w$ of $\mathcal{W}_{a,b}$, there
is some unique $(a^{'},b^{'})$ and some unique $w^{'}\in\mathcal{W}_{a^{'},b^{'}}$.
Thus, it is possible to group the terms of the summation in Equation
\ref{eq:fmdproof} by $(a^{'},b^{'})$. Now, the corresponding $w^{'}$
will have one step fewer, so $l(w^{'})+1=l\left(w\right)$, and it
omits the transition probability for the step from $\left(a,b\right)$
to $(a^{'},b^{'})$, so 
\[
p\left(w\right)=p\left((a^{'},b^{'})\,\middle|\,\left(a,b\right)\right)p\left(w^{'}\right)
\]

Thus, Equation \ref{eq:fmdproof} is rewritten as:
\begin{multline*}
s\left(a,b\right)=\\
\sum_{(a^{'},b^{'})\in\mathcal{R}\left(\left(a,b\right)\right)}\sum_{w^{'}\in\mathcal{W}_{a^{'},b^{'}}}p\left((a^{'},b^{'})\,\middle|\,\left(a,b\right)\right)p\left(w^{'}\right)C^{l\left(w^{'}\right)+1}\\
=C\sum_{(a^{'},b^{'})\in\mathcal{R}\left(\left(a,b\right)\right)}p\left((a^{'},b^{'})\,\middle|\,\left(a,b\right)\right)\sum_{w^{'}\in\mathcal{W}_{a^{'},b^{'}}}p\left(w^{'}\right)C^{l\left(w^{'}\right)}
\end{multline*}

But, by definition, $s(a^{'},b^{'})=\sum_{w^{'}\in\mathcal{W}_{a^{'},b^{'}}}p\left(w^{'}\right)C^{l\left(w^{'}\right)}$.

This completes the proof that the expected $f$-meeting distances
satisfy the recursive form of Equation \ref{eq:general_rec}. By the
uniqueness result of Theorem \ref{thm:uniq}, it follows that this
is the same as the solution that can be arrived as a limit of the
iterative form, thus establishing the equivalence of the Generalized
Random Surfer-Pair model and its recursive form.
\end{proof}
This is the result necessary to convert an existing recursive form
into a Random Surfer-Pair Model. All that needs to be done is to read
off the non-zero coefficients into the appropriate places into the
matrix, or equivalently get the support set and the corresponding
transition probabilities as a function of $\left(a,b\right)$. One
thing to note here is that the probabilities corresponding to \emph{actual}
node pair destinations from $\left(a,b\right)$ need not sum to 1,
because it could go into the stopped state as well.

It would of course be more illustrative to get a concise description
of the matrix. For SimRank for instance, $\mathcal{R}\left(\left(a,b\right)\right)$
is simply all neighbor pairs of $a$ and $b$ with the transition
probabilities being uniform over this set. If either node has no in-neighbors,
it transitions to the stopped state with probability 1 (i.e unavailable).

Our formulation allows for transitions from one compound state to
any arbitrary state, and any number of destination states with non-zero
transition probability. However, in existing measures, there are only
a few possible transitions from any given state $\left(a,b\right)$
compared to the total number of possible states, that too involving
the neighbor pairs of $a$ and $b$ (on practical graphs which are
not densely connected). This means that $\mathbf{P}$ is sparse, so
even though there are $\left|V\right|^{2}$ states, only a few of
them are involved in transitions from any given state, and it is no
more complicated than the existing recursive formulations. However,
we note that the above results continue to hold for any $\mathbf{P}$
regardless of sparsity.

\section{Application to Other Measures}

In this section, we demonstrate the usage of the GRSP model by applying
it to PSimRank (Equation \ref{eq:psim_rec}), P-Rank (Equation \ref{eq:prank_rec})
and SimRank{*} (Equation \ref{eq:simrank-star-rec}) as discussed
in the previous section. Throughout this exercise, any ``unallocatable''
probability due to e.g unavailability of in-neighbors or out-neighbors
is given to the stopped state.

\subsection{PSimRank\label{subsec:PSimRank-grsp}}

\paragraph*{Support Set}

$\mathcal{R}\left(\left(a,b\right)\right)$ is divided into 3 disjoint
subsets: $\left\{ \left(x,x\right),\,x\in I\left(a\right)\cap I\left(b\right)\right\} $,
$\left(I\left(a\right)\setminus I\left(b\right)\right)\times I\left(b\right)$,
$I\left(a\right)\times\left(I\left(b\right)\setminus I\left(a\right)\right)$.
Note that parts of $I\left(a\right)\times I\left(b\right)$ are missing
from this set unlike SimRank.

\paragraph*{Transition Probabilities}

The 3 subsets mentioned above are given total probabilities of $\frac{|I\left(a\right)\cap I\left(b\right)|}{|I\left(a\right)\cup I\left(b\right)|}$
(note that this is the the Jaccard coefficients of in-neighbor sets),
$\frac{|I\left(a\right)\setminus I\left(b\right)|}{|I\left(a\right)\cup I\left(b\right)|}$
and $\frac{|I\left(b\right)\setminus I\left(a\right)|}{|I\left(a\right)\cup I\left(b\right)|}$
respectively, and probabilities for individual elements are uniform
over each set. This is exactly the RSP model presented in \cite{psimrank}.

\subsection{P-Rank}

\paragraph*{Support Set}

$\mathcal{R}\left(\left(a,b\right)\right)$ now has two disjoint parts:
$I\left(a\right)\times I\left(b\right)$ and $O\left(a\right)\times O\left(b\right)$.

\paragraph*{Transition Probabilities}

The two parts of $\mathcal{R}\left(\left(a,b\right)\right)$ are given
total probabilities of $\lambda$ and $1-\lambda$, and just like
PSimRank, probabilities for individual elements are uniform over each
set. These probabilities can be interpreted as follows: a coin with
probability $\lambda$ is tossed, and based on its result, \emph{both}
surfers move backward or forward and choose from applicable edges
uniformly. This scheme is illustrated in Figure \ref{fig:prank-grsp}.
Transitions where the surfers take different directions are not allowed.
Thus, a theoretical deficiency exists that causes it to discard path-pairs
containing such transitions as evidence for similarity. This is not
at all evident from its recursive form.

\begin{figure}[t]
\begin{centering}
\includegraphics[width=0.25\textwidth]{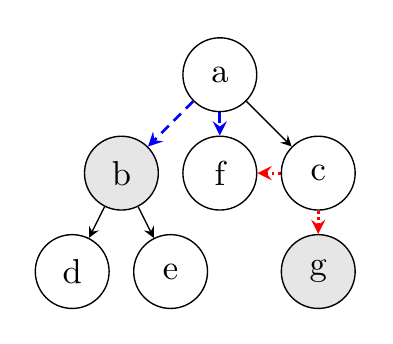}\includegraphics[width=0.25\textwidth]{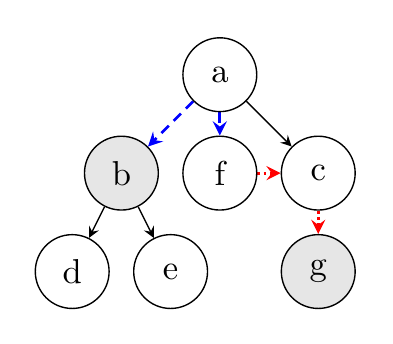}
\par\end{centering}
\caption{Left: An instance of the GRSP experiment for P-Rank with the surfers
starting at $b$ (dashed path) and $g$ (dotted path) and meeting
at $f$. Right: In this graph, the same path-pair is invalid because
the paths take opposite directions in the second step.\label{fig:prank-grsp}}

\end{figure}

\subsection{SimRank{*}}

\paragraph*{Support Set}

For SimRank{*}, $\mathcal{R}\left(\left(a,b\right)\right)=$$\left(\left\{ a\right\} \times I\left(b\right)\right)\cup\left(I\left(a\right)\times\left\{ b\right\} \right)$.
Note that the Cartesian products involve singleton sets.

\paragraph*{Transition Probabilities}

The two parts of $\mathcal{R}\left(\left(a,b\right)\right)$ are given
total probabilities of $\frac{1}{2}$ each, and again, individual
probabilities are uniform over their respective subsets. These transitions
are the same as tossing a fair coin, and based on the outcome, stepping
\emph{one} surfer to a uniformly chosen in-neighbor. An example of
this is shown in Figure \ref{fig:simstar}.

The notable feature here is that only one of the surfers is allowed
to move at each step. The choice as to which surfer moves is made
uniformly. From this, it becomes clear how SimRank{*} manages to consider
path-pairs of unequal length. The surfers need not have made an equal
number of jumps to meet at some node. This is a much more simpler
and intuitive explanation than the original analytic derivation in
\cite{simrank-star} by enumerating all such path pairs.

\begin{figure}[t]
\begin{centering}
\includegraphics[width=0.2\textwidth]{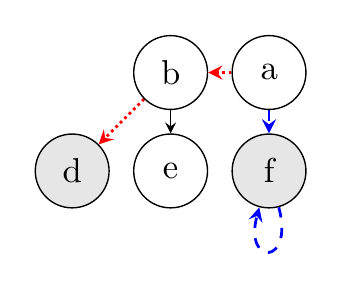}
\par\end{centering}
\caption{An instance of the GRSP experiment for SimRank{*} with the surfers
starting at $d$ (dotted path) and $f$ (dashed path) and meeting
at $a$, with their paths having different effective lengths. The
loop shows that surfer did not move that step.\label{fig:simstar}}
\end{figure}

\section{Monte Carlo Computation\label{sec:MC-Computation}}

A key benefit of the GRSP model is that it enables the use of Monte
Carlo methods for the entire class of measures that fall under this
model. This includes SimRank{*} and P-Rank, for which Monte Carlo
methods have not been used so far. It is quite straightforward to
apply: simulate the surfers' transitions starting from a given pair
of nodes for some number of times, and return the average score as
the similarity.

In practice, the surfers would have to be truncated after some number
of steps $L_{max}$, and only a limited number of samples $N_{S}$
can be drawn in the interest of fast querying, but decent guarantees
for accuracy in practice are shown in \cite{psimrank}. Radius based
pruning can be done to reduce the amount of nodes that need to be
considered for top-$k$ similarity queries, i.e restricting the search
to nodes at a particular distance (or radius) from the query node.

The various benefits of using Monte Carlo computation are as follows:
\begin{enumerate}
\item \textit{Complexity:} Typically, there is an easy ($O(1)$) way to
generate a transition from any given state. In SimRank{*} for instance,
all that needs to be done is to toss a coin, and advance one surfer
to a randomly chosen in-neighbor. Thus, the complexity of a single
similarity computation is just $O(N_{S}L_{max})$, where these quantities
are much smaller than the size of the network. In comparison, solving
the recursive equations takes at least $O(\left|V\right|^{2})$ even
when extensive optimizations can be made \cite{simrank-star}.
\item \textit{Memory efficiency:} The memory requirement is also low since
individual similarities can be computed without having to compute
and store all pairs of similarities ($O(N^{2})$ memory needed) as
is necessary for the recursive forms. More precisely, it is $O(1)$
for each similarity being queried. This is particularly important
when dealing with large networks (including citation networks). In
fact, even storing the iterates while solving the recursive forms
becomes infeasible for even medium sized networks (of the order of
$10^{6}$ nodes) by current standards.
\item \textit{Adaptability to changes in network structure:} The network
can be updated without having to recompute all similarities, which
is useful in settings where nodes are dynamically added or removed.
This is often the case, such as with social and citation networks,
which change on a regular basis. Solving the recursive forms after
each change would be highly impractical or even impossible.
\item \textit{Parallelizability:} Possibly the biggest advantage is the
extensive parallelizability; every instance of the RSP experiment
can be run separately and concurrently. Further, under the RSP model,
different similarities can be computed in parallel.
\end{enumerate}

\section{Designing Measures Under the GRSP Model}

Theorem \ref{thm:equivalence}, along with our other results, characterizes
a class of ``sensible'' measures that are non-negative, satisfy
the base case of maximal self similarity being 1, and are bounded.
It describes how such measures in the recursive form weight common
sources of similarity in their computation in their GRSP model. We
have also seen how existing measures can be better understood under
the GRSP model, and the computational benefits that result. These
observations highlight the utility of the GRSP model as a tool for
developing measures according to given design requirements, rather
than working with the recursive forms directly.

\subsection{Modes of Design}

The only mathematical requirement is a valid transition matrix $\mathbf{P}$,
i.e $\mathbf{P}$ needs to be doubly stochastic. Thus, designing measures
is a matter of allocating transition probabilities out of each compound
state. This could be done from first principles such as those underlying
the measures we have discussed so far, or existing measures could
be involved, i.e $\mathbf{P}$ could be derived from some other transition
matrices. For example, one might wish to combine existing measures
in order to create a measure with the desirable properties present
in each of the existing measures, an exercise we undertake in Section
\ref{subsec:psimstar}.

One natural way to combine measures is to take a convex combination
of the individual transition matrices, resulting in a transition matrix
of the form
\[
\mathbf{P}=\sum_{i=1}^{N}\lambda_{i}\mathbf{P}_{i}
\]
with $\mathbf{P}_{i},\,i=1\ldots N$ being transition matrices for
some given measures, and $\lambda_{i}$ are non-negative weights such
that $\sum_{i}\lambda_{i}=1$. In fact, it is easy to see that P-Rank
itself is such a combination. The two measures involved are SimRank
and Rvs-SimRank, as can be verified by using Theorem \ref{thm:equivalence}
to get the transition matrices from Equations \ref{eq:sr_rec}, \ref{eq:rvs_sr_rec}
and \ref{eq:prank_rec}. Indeed, the rationale behind P-Rank is to
account for both in-links as well as out-links, which are properties
of these two measures that are easily combined as described above.

\subsection{Domain Knowledge}

The GRSP framework also enables the use of domain knowledge in the
design process. We describe two possible ways in which this can be
accomplished:
\begin{itemize}
\item Knowledge about the network structure: Different parts of the network
could have different structures, and a mixture of behaviors of various
measures could be necessary to accurately capture similarity. One
possibility is to construct $\mathbf{P}$ based on independent transitions
for each component of the node pair following different measures based
on which part of the network the node belongs to, leading to probabilities
of the form
\[
p\left((a^{'},b^{'})\,\middle|\,\left(a,b\right)\right)=p_{1}\left(a^{'}\,\middle|\,a\right)p_{2}\left(b^{'}\,\middle|\,b\right)
\]
where the RHS terms are obtained by appropriately marginalizing from
two given transition matrices $\mathbf{P}_{1}$ and $\mathbf{P}_{2}$.
\item Node and Edge attributes and other data: these can be used to give
more weight to meaningful transitions that can be identified from
the additional information. Similarly, this could help identify and
prune out links that exist in the network but are not indicative of
similarity in any way. This could prove particularly beneficial when
utilizing features extracted from text data. For instance, \cite{cit-intens}
proposes a method to detect the intensity of references in scientific
articles based on their textual content, that is how important a reference
is to an article. This inferred attribute could be used to assign
lower probabilities to tangential references when computing similarity
as papers connected through such references can be very dissimilar.
\end{itemize}

\subsection{PSimRank{*} : The Best of Both Worlds\label{subsec:psimstar}}

In this section, we describe a measure designed based on two existing
measures using the GRSP framework. Later in section \ref{sec:Experiments},
the efficacy of this measure is empirically evaluated.

Previously, we have described how PSimRank solves the Pairwise Normalization
problem and SimRank{*} solves the Level-wise Computation problem,
and what these entail in the Random Surfer-Pair domain. With PSimRank{*},
we attempt to combine these two benefits under the GRSP framework,
resulting in a better measure because of solving both the problems.

The combination is straightforward; to make the surfers meet at a
common in-neighbor with probability equal to the Jaccard coefficient
just like in PSimRank, but the remainder of the time behave like SimRank{*},
moving only one at a time.

\paragraph*{Support Set}

$\mathcal{R}\left(\left(a,b\right)\right)=$$\left(\left\{ a\right\} \times I\left(b\right)\right)\cup\left(I\left(a\right)\times\left\{ b\right\} \right)\cup\left\{ \left(x,x\right),\,x\in I\left(a\right)\cap I\left(b\right)\right\} $.
The first two subsets are for behavior like SimRank{*}, and the third
subset is from PSimRank with both surfers stepping to a common in-neighbor.

\paragraph*{Transition Probabilities}

The third subset is allocated probability $\frac{|I\left(a\right)\cap I\left(b\right)|}{|I\left(a\right)\cup I\left(b\right)|}$,
and the remaining is divided equally among the other two. The same
scheme as before is used for individual probabilities.

The effect of this design is that node pairs which have a large number
of common neighbors (i.e a high value for the Jaccard coefficient
$\frac{|I\left(a\right)\cap I\left(b\right)|}{|I\left(a\right)\cup I\left(b\right)|}$)
will greatly increase the tendency of the measure to adopt PSimRank-like
behavior and jump to a common neighbor, thus solving the Pairwise
Normalization problem for these nodes. For other node pairs with fewer
common neighbors, this problem is not as severe, and the measure will
focus more on the Level-wise Computation problem by adopting SimRank{*}-like
behavior. This way, PSimRank{*} is expected to mitigate both problems
overall.

\section{Experiments\label{sec:Experiments}}

We compare the performance of PSimRank{*} against existing measures
on a real-world dataset. For P-Rank, a sweep over the $\lambda$ parameter
is performed in increments of 0.1 from 0 to 1 and the best performing
value of $\lambda$ is used for comparison. SimRank and rvs-SimRank
are the edge cases of this sweep with $\lambda=1$ and $\lambda=0$
respectively.

We use the Arnetminer dataset (\cite{arnetminer}), which is a citation
network of 2,244,021 papers and 4,354,534 citations extracted from
DBLP\footnote{\href{http://dblp.uni-trier.de/}{http://dblp.uni-trier.de/}}.
A portion of these papers have been manually annotated and given labels
corresponding to 10 different topics (clusters). The evaluation consists
of running a top-k similarity query on some randomly chosen labeled
nodes, and finding the Mean Average Precision (MAP) (\cite{nlpbook})
for the answer set having the same label as the query. The rationale
behind this is that papers in the same topic as the query are likely
to be similar.

For all measures used here, pruning was done to radius 4 (in the undirected
graph). The Random Surfer simulations were performed 200 times per
query, and truncated after at most 15 steps. Top-100 queries were
run on 50 randomly chosen labeled nodes that had at least 5 citations
and 5 references to ensure that the measures wouldn't become unavailable.
Since not all the nodes are labeled, only the nodes in the answer
set that have a label are considered for calculating the MAP scores.
Further, 50 such trials are performed and the averaged MAP scores
are reported in Table \ref{tab:sim-results}.

\begin{table}[tb]
\caption{MAP values attained by various measures on different datasets.\label{tab:sim-results}}
\centering{}%
\begin{tabular}{|c|c|c|c|}
\hline 
Measure & Arnetminer & Citeseer & Cora\tabularnewline
\hline 
\hline 
SimRank & 0.73 & 0.71 & 0.66\tabularnewline
\hline 
P-Rank($\lambda$=0.4) & 0.76 & 0.73 & 0.70\tabularnewline
\hline 
SimRank{*} & 0.80 & 0.67 & 0.62\tabularnewline
\hline 
PSimRank & 0.80 & 0.68 & 0.57\tabularnewline
\hline 
PSimRank{*} & 0.81 & 0.69 & 0.63\tabularnewline
\hline 
\end{tabular}
\end{table}

It is observed that PSimRank{*} outperforms all the other measures,
improving on both PSimRank and SimRank{*} on which it is based. In
networks much smaller than Arnetminer such as preprocessed versions
\cite{cora-reduced} of the Cora \cite{cora-orig} and CiteSeer \cite{citeseer-orig}
datasets which have only a few thousand nodes each, we found that
PSimRank{*} as well as SimRank{*} and PSimRank performed worse than
P-Rank (Table \ref{tab:sim-results}). However, even in these cases,
PSimRank{*} outperformed its predecessors. Thus, combining the benefits
of PSimRank and SimRank{*} under the GRSP interpretation is indeed
effective, improving on both the measures.

\section{Conclusions}

The GRSP model serves as a unifying framework for a class of similarity
measures based on the SimRank philosophy and subsumes several seemingly
disparate measures. Any properties that are discovered for this framework
would also apply to these measures. Admittedly, it is not all-encompassing;
it is not evident how it can be applied to MatchSim \cite{matchsim}
which uses weights obtained from a Maximum-Matching based scheme (which
are not constant), and CoSimRank \cite{cosimrank} which is based
on Personalized PageRank \cite{ppr}.

Reinterpreting existing measures (P-Rank and SimRank{*}) under this
framework has provided interesting insights about their functioning.
The benefits of Monte Carlo computation are also brought to this class
of measures. The development of PSimRank{*} under this framework,
improving on the measures from which it was derived, highlights the
potential of the GRSP Model to aid in the designing measures tailored
to various applications and domains. One exciting avenue for future
work is to use this framework to incorporate knowledge generated by
Machine Learning methods, such as document embeddings generated by
the state of the art Natural Language Processing methods. Hopefully,
this work has opened up possibilities for theoretical dissection and
development of more effective measures.

\section*{Acknowledgements}

The authors thank the Robert Bosch Centre for Data Science and Artificial
Intelligence, IIT Madras for providing the computational resources
necessary for this paper.

\bibliographystyle{plain}
\bibliography{refs}

\end{document}